\documentclass[12pt]{article}
\usepackage{graphicx}
\usepackage{amssymb}
\usepackage[latin1]{inputenc}
\usepackage{Qcircuit}
\usepackage{amsmath}
\usepackage{braket}

\newcommand{\newc}{\newcommand}
\newc{\beq}{\begin{equation}}
\newc{\eeq}{\end{equation}}
\newc{\bea}{\begin{array}}
\newc{\eea}{\end{array}}
\newcommand{\ben}{\begin{eqnarray}}
\newcommand{\een}{\end{eqnarray}}
\newc{\ra}{\rightarrow}
\newc{\bfx}{{\bf x}}
\newc{\bfV}{{\bf V}}
\newc{\cO}{{\cal O}}
\newc{\bfv}{{\bf v}}
\newc{\bfu}{{\bf u}}
\newc{\bfp}{{\bf p}}
\newc{\ve}{{\varepsilon}}
\newc{\Psibar}{\overline\Psi}
\newc{\w}{{\bf w}}
\newc{\E}{{\mathbf{E}}}
\newc{\EE}{{\mathcal E}}
\newc{\bfn}{{\mathbf\nabla}}
\newc{\la}{{\cal L}}
\newc{\tla}{{\tilde{\cal L}}}
\newc{\bp}{{\bf p}}
\newc{\ho}{\hookrightarrow }
\newc{\bP}{{\bf P}}
\newc{\pd}{{\partial}}
\newc{\piv}{{\partial_4}}
\newc{\pv}{{\partial_5}}
\newc{\bJ}{{\bf J}}
\newc{\bze}{{\mathbf 0}}
\newc{\bK}{{\bf K}}
\newc{\tphi}{{\tilde\phi}}
\newc{\tF}{{\tilde F}}
\newc{\tD}{{\tilde D}}
\newc{\tJ}{{\tilde J}}
\newc{\tj}{{\tilde j}}
\newc{\bD}{{\bf D}}
\newc{\tvphi}{{\tilde\varphi}}
\newc{\trho}{{\tilde\rho}}
\newc{\ttheta}{{\tilde\theta}}
\newc{\tpsi}{{\tilde\psi}}
\newc{\tu}{{\tilde u}}
\newc{\cD}{{\cal D}}
\newc{\tPhi}{{\tilde\Phi}}
\newc{\tPsi}{{\tilde\Psi}}
\newc{\tA}{{\tilde A}}
\newc{\talpha}{{\tilde\alpha}}
\newc{\tbeta}{{\tilde\beta}}
\newc{\bA}{{\mathbf A}}
\newc{\bB}{{\bf B}}
\newc{\br}{{\bf r}}
\newc{\sig}{{\mathbf\sigma}}
\newc{\eg}{{\rm e.g.\ }}
\newc{\ie}{{\rm i.e.\ }}
\newcommand{\bey}{\begin{eqnarray}}
\newcommand{\pslash}{\not{\hbox{\kern-2.3pt $p$}}}
\newcommand{\pdslash}{\not{\hbox{\kern-2pt $\partial$}}}
\newcommand{\eey}{\end{eqnarray}}
\newtheorem{theorem}{Theorem}
\newtheorem{proposition}[theorem]{Proposition}
\newenvironment{proof}[1][Proof]{\noindent\textbf{#1.} }{\ \rule{0.5em}{0.5em}}
\newtheorem{lemma}{Lemma}
\newtheorem{example}[theorem]{Example}
\newtheorem{definition}[theorem]{Definition}
\begin{document}

\begin{titlepage}
\vskip 2cm
\begin{center}
{\Large  Clifford Algebras, Quantum Neural Networks and Generalized Quantum Fourier Transform
\footnote{{\tt matrindade@uneb.br}}}
 \vskip 10pt
{ Marco A. S. Trindade$^{\dag}$  \\
Vinicius N. L. Rocha $^{\dag \dag}$ \\
S. Floquet $^{\dag \dag \dag}$ \\}
\vskip 5pt
{\sl $^{\dag}$Colegiado de Física, Departamento de Ciências Exatas e da Terra, Universidade do Estado da Bahia\\
Salvador, Bahia, Brazil \\
$^{\dag \dag}$ Atos, Latin American Quantum Computing Center, Salvador-BA, 41650-010, Brazil \\
$^{\dag \dag \dag}$ Colegiado de Engenharia Civil, Universidade Federal do Vale do São Francisco, Juazeiro, Bahia, Brazil}
\vskip 2pt
\end{center}

\begin{abstract}

We propose models of quantum neural networks through Clifford algebras, which are capable of capturing geometric features of systems and to produce entanglement. Due to their representations in terms of Pauli matrices, the Clifford algebras are the natural framework for multidimensional data analysis in a quantum setting. Implementation of activation functions and unitary learning rules are discussed. In this scheme, we also provide an algebraic generalization of the quantum Fourier transform containing additional parameters that allow performing quantum machine learning. Furthermore, some interesting properties of the generalized quantum Fourier transform have been proved.
\end{abstract}

\bigskip


\vskip 3pt

\end{titlepage}


\newpage

\setcounter{footnote}{0} \setcounter{page}{1} \setcounter{section}{0} %
\setcounter{subsection}{0} \setcounter{subsubsection}{0}
\section{Introduction}
Artificial neural networks are computing models made up of elementary units called artificial neurons. There are several important applications in many fields such as pattern recognition \cite{Image}, medical diagnosis \cite{Med1, Med2}, combinatorial optimization problems \cite{Opt}. With the advent of quantum computing quantum neural networks have been proposed. The seminal work on quantum neural computing was developed by Kak \cite{Kak}. Posteriorly, Altaisky \cite{Alt} proposes a quantum version of perceptron. However the learning algorithm is not unitary, one of crucial ingredients of quantum mechanics. Since then, many approaches have appeared in literature. Silva \cite{Ade} \emph{et al} defined a quantum perceptron over a field in order to overcome the limitations of high cost learning algorithms in classical neural networks. It was proposed a superposition based architecture learning algorithm (SAL) to optimize the weights of a neural network. In reference \cite{Beer}, Beer \emph{et al} introduced quantum deep neural networks with an efficient quantum training algorithm, using the fidelity as a cost function. In this scheme, tolerance to noise training data has been demonstrated. Another interesting model was proposed by Shao \cite{Shao}. It consists of a quantum feed-forward neural networks whose learning algorithm is unitary and it contains quantum superposition and parallelism features. In addition, the Hadamard and swap tests were explored. The procedure is analogous to variational quantum eigensolvers.

Variational quantum cicuits, which are also sometimes quantum neural networks can be explored as quantum machine learning models \cite{Bene, Maria}. Due to limitations of near-term quantum computing, variational circuits have provided new perspectives that exceed the questions related to computational speedups \cite{Harrow, Lloyd} (which are obviously extremely relevant and constitute the gold standard of algorithmic design \cite{Maria}). Thus, the usefulness of quantum properties as superposition and entanglement \cite{Cai} can be investigated in this context \cite{Maria}. As highlighted by \cite{Biamonte} quantum systems can be generate patterns in data which are not feasible for classical systems. Besides, quantum machine learning may be able to recognize and classify patterns that are classically inaccessible \cite{Biamonte}. Possible guides for the construction of more complex architectures that capture non-classical features and geometric and topological information are geometric algebras (Clifford algebras).

Neural network model based on quantum information processing and quaternions (even subalgebras of Clifford algebra $Cl(3,0)$ has been proposed by Teguri \emph{et al} \cite{Teguri} called QQNN (Quaternionic Qubit Neural Network). Numerical experiments indicate a better performance in the prediction time-series of a chaotic system, comparated to conventional real-valued network. Buchholz and Sommer \cite{Sommer, Sven} have analyzed Clifford algebras for the design of neural architectures capable of processing a plethora of geometric objects. Corrochano \emph{et al} \cite{Ed}  formulated a quanternionic quantum neural network with applications for pattern recognition. It was shown that the QQNN have a good performance than others approaches, requiring less inputs per pattern (APM- autonomous perceptron models). Quaternions can be identified as the even part of Clifford algebra $Cl(3,0)$. There are several applications of Clifford algebras in quantum computation. Vlasov \cite{Vlasov} used Clifford algebras $Cl(2n,0)$ on the construction of universal sets of quantum gates. Josza and Miyake has explored Jordan-Wigner representations of Clifford algebras in Gaussian quantum circuits \cite{Josza} and the computational complexity of Grover algorithm can be simplified with Clifford algebras \cite{Alves}. Trindade \emph{et al} \cite{Trindade} developed a formalism based on Clifford algebras for decoherence-free subspaces, a special class of quantum error correcting codes. Still from an algebraic point of view, symmetry groups have recently been investigated in quantum Boltzmann machines \cite{Song}

In this work, we analyze the underlying algebraic structure of quantum neural networks. From this point of view, we propose a quantum neural network based on Clifford algebras $Cl(2n,0)$ and $Cl(3,0)^{\otimes n}$. In Section 2 we present our models of the quantum neural networks and error analysis based on results of simulation Hamiltonian. Section 3 contains a algebraic generalization of quantum Fourier transform. Section 4 is devoted to the conclusions and perspectives. In Appendix A we review some basic concepts of Clifford algebras and in the Appendix B, we provide an example of a simple circuit based on a representation of Clifford algebra.

\section{Formulation}
Initially, our aim is to highlight that we can build arbitrary unitary matrices from Clifford algebras. The next lemma contains the analogous results obtained by Vlasov \cite{Vlasov} and Winter \cite{Winter}, however we use the reversion operator in a Clifford algebra.
\begin{lemma}
Let $Cl(2n,0) \equiv Cl(2n)$ be a Clifford algebra. There is a representation of this algebra whose basis induces a basis of vector space of Hermitian matrices $2^{n} \times 2^{n}$ over $\mathbb{R}$.
\end{lemma}
\begin{proof}
Consider the following representation \cite{Vlasov}
\begin{eqnarray}\label{Vlas}
\Gamma_{2k}&=&\underbrace{I \otimes \cdots I}_{n-k-1} \otimes \  \sigma_{x} \otimes \underbrace{\sigma_{z} \otimes \cdots \otimes \sigma_{z}}_{k}, \nonumber   \\
\Gamma_{2k+1}&=&\underbrace{I \otimes \cdots I}_{n-k-1} \otimes \ \sigma_{y} \otimes \underbrace{\sigma_{z} \otimes \cdots \otimes \sigma_{z}}_{k},
\end{eqnarray}
with $k=0,1,...,n-1$. These $2n$ matrices (generators of $Cl(2n)$) are linearly independent and Hermitian. Now we consider the $2^{2n}$ operators \cite{Winter}
\begin{eqnarray}
1, \ \Gamma_{j_1}, \ i \Gamma_{j_1} \Gamma_{j_2}, \  i \Gamma_{j_1} \Gamma_{j_2} \Gamma_{j_3}, \ \cdots, \ \omega  \Gamma_{j_1} \Gamma_{j_2}...\Gamma_{j_s},...,\  \cdots, \ \omega \Gamma_{j_1} \Gamma_{j_2},...\Gamma_{2n},
\end{eqnarray}
where $0 \leq j_1 \leq j_2 \leq \cdots \leq 2n-1$ and $\omega=i$ if $\widetilde{\Gamma}_{j_1, \cdots, j_{\zeta}}=(-1)^{\zeta(\zeta -1)/2}\Gamma_{j_1, \cdots, j_{\zeta}}=(-1)\Gamma_{j_1, \cdots, j_{\zeta}}$ and $\omega=1$ if $\widetilde{\Gamma}_{j_1, \cdots, j_{\zeta}}=(-1)^{\zeta(\zeta -1)/2}\Gamma_{j_1, \cdots, j_{\zeta}}=\Gamma_{j_1, \cdots, j_{\zeta}}$. The insertion of $i$ ensures that matrices are Hermitians. Hence using the representation (\ref{Vlas}) we have $2^{2n}$ linearly independent Hermitian matrices, i.e., a basis for the vector space of Hermitian matrices of dimension $2^{n} \times 2^{n}$.
\end{proof}

\begin{example}
For the $Cl(2)$, we have $I, \Gamma_{0}=\sigma_{x}, \Gamma_{1}=\sigma_{y}, \Gamma_{2}=i\sigma_{x}\sigma_{y}$.
\end{example}

\begin{example}
For the $Cl(2^{2})=Cl(4)$, the elements of basis are $1, \Gamma_{0}=I \otimes \sigma_{x}, \Gamma_{1}=I \otimes \sigma_{y}, \Gamma_{2}=\sigma_{x} \otimes \sigma_{z}, \Gamma_{3}=\sigma_{y} \otimes \sigma_{z}, i \Gamma_{0}\Gamma_{1}=i(I \otimes \sigma_{x}\sigma_{y}), i \Gamma_{0}\Gamma_{2}=i(\sigma_{x} \otimes \sigma_{x}\sigma_{z}), i \Gamma_{0}\Gamma_{3}=i(\sigma_{y} \otimes \sigma_{x}\sigma_{z}), i \Gamma_{1}\Gamma_{2}=i(\sigma_{x} \otimes \sigma_{y}\sigma_{z}), i \Gamma_{1}\Gamma_{3}=i(\sigma_{y} \otimes \sigma_{y}\sigma_{z}), i \Gamma_{2}\Gamma_{3}=i(\sigma_{x}\sigma_{y} \otimes I), i \Gamma_{0}\Gamma_{1}\Gamma_{2}=i(\sigma_{x} \otimes \sigma_{x}\sigma_{y}\sigma_{z}),
i \Gamma_{0}\Gamma_{2}\Gamma_{3}=i(\sigma_{y} \otimes \sigma_{x}\sigma_{y}\sigma_{z}), i \Gamma_{1}\Gamma_{2}\Gamma_{3}=i(\sigma_{x}\sigma_{y} \otimes \sigma_{y}),  \Gamma_{0}\Gamma_{1}\Gamma_{2}\Gamma_{3}=\sigma_{x} \sigma_{y} \otimes \sigma_{x}\sigma_{y}$.
\end{example}

Elements of Clifford algebras carry a natural geometric interpretation. The scalars and vectors have the standard interpretation. Bivectors can represent an oriented area or oriented angle associated to rotation and trivectors have an interpretation as an oriented volume. Furthermore, Clifford algebras generalize hypercomplex numbers, including real numbers, complex numbers and quaternions. These algebraic structures can be used in neural network architectures. \ \

Clifford quantum architectures match entanglement with geometric information of data. In the next definition, we can encode $2^{2n}$ neurons (coefficients $\alpha$) into $n$ qubits. This formulation is similar to the reference \cite{Shao}. However, our approach is based on Clifford algebras.

\begin{definition}
A Type I Clifford quantum perceptron (CQP-Type I) is defined by $(\vert x \rangle, \vert w \rangle, \vert y \rangle)$ where the input $\vert x \rangle$ is given by
\begin{eqnarray}
\vert x \rangle = \exp\left(i\sum_{j=0}^{2^{n}-1} \omega_{j} \alpha_{j} \Gamma_{j}^{(V)}\right) \vert 0 \rangle^{\otimes n},
\end{eqnarray}
The weight is
\begin{eqnarray}
\vert w \rangle = \exp\left(i\sum_{j=0}^{2^{n}-1} \omega_{j} \theta_{j} \Gamma_{j}^{(V)}\right) \vert 0 \rangle^{\otimes n}
\end{eqnarray}
and the output defined as
\begin{eqnarray}
\vert y \rangle = \exp(i  \omega_{j}  \phi \Gamma_{\mu}^{(V)}) \vert 0 \rangle^{\otimes n}.
\end{eqnarray}
with $\phi=\varphi (\langle x \vert w \rangle)$, where $\varphi$ is the activation function and $\Gamma_{\mu}^{(V)}$ stands for an arbitrary element of basis for the Clifford algebra $Cl(2n)$ and $\omega_{j}$ is defined in the Lemma 1.
\end{definition}

Note that these gates can generate entangled states since they are not necessarily factorable. This definition may not be interesting if all coefficients are non-zero $\alpha_{j}, \theta_{j}, \phi$ because the number of quantum logical gates to be applied grows exponentially with n, unless some approximation is made which we will discuss later. A better choice for polynomial growth is given by following definition, which is a particular case of Definition 3.

\begin{definition}
 A type II Clifford quantum perceptron (CQP-Type II) is defined by $(\vert x \rangle, \vert w \rangle, \vert y \rangle)$ where the input $\vert x \rangle$ is given by
\begin{eqnarray}
\vert x \rangle = \exp\left(i\sum_{j=0}^{2n-1}  \alpha_{j} \Gamma_{j}\right) \vert 0 \rangle^{\otimes n},
\end{eqnarray}
the weight is defined as
\begin{eqnarray}
\vert w \rangle = \exp\left(i\sum_{j=0}^{2n-1}  \theta_{j} \Gamma_{j}\right) \vert 0 \rangle^{\otimes n}
\end{eqnarray}
and the output as
\begin{eqnarray}
\vert y \rangle = \exp(i  \phi \Gamma_{\mu}) \vert 0 \rangle^{\otimes n}.
\end{eqnarray}
with $\phi=\varphi (\langle x \vert w \rangle)$, where $\varphi$ is the activation function.
\end{definition}

Here we encode $2n$ neurons into $n$ qubits. This scheme can be generalized into multilayer case. Consider the parameter $\theta_{m,j_{m}}$, where the index $m$ refers to the mth-layer, $j_{m}$ are the indices of neuron in the mth-layer so that we have
\begin{eqnarray}
\vert x \rangle_{m-1} &=& \exp\left(i\sum_{j=0}^{2n-1} \omega_{j} \alpha_{j,m-1} \Gamma_{j}\right) \vert 0 \rangle^{\otimes n},  \\
\vert w \rangle_{m, j_{m}} &=& \exp \left(i\sum_{j=0}^{2n-1} \omega_{j} \theta_{m,j_{m}} \Gamma_{j}\right) \vert 0 \rangle^{\otimes n},  \\
\vert y \rangle_{m} &=& \exp \left(i \sum_{j=0}^{2n-1} \omega_{j} \phi_{j,m} \Gamma_{j}\right) \vert 0 \rangle^{\otimes n}.
\end{eqnarray}
where $\phi_{j,m}=\arccos(\phi(_{m-1}\langle x \vert w \rangle_{m, j_{m}}))$. \ \

For the learning algorithm, we consider the training sample in the CQP-Type I
\begin{eqnarray}
\vert x^{(s)} \rangle &=& \exp\left(i\sum_{j=0}^{2n} \omega_{j} \alpha_{j}^{(s)} \Gamma_{j}^{(V)}\right) \vert 0 \rangle^{\otimes n}, \nonumber \\
\vert r^{(s)} \rangle &=& \exp\left(i  \beta^{(s)} \Gamma_{\mu}^{(V)} \right) \vert 0 \rangle^{\otimes n},
\end{eqnarray}
where $\vert r^{(s)} \rangle$ is the desired output. We have that
\begin{eqnarray}
\vert y^{(s)} \rangle &=& \exp\left(i  \phi^{(s)} \Gamma_{\mu}^{(V)}\right) \vert 0 \rangle^{\otimes n},
\end{eqnarray}
where $\phi^{(s)}=\arccos \varphi (\langle x^{(s)} \vert w \rangle)$. Thus, the quantum fidelity can be expressed as
\begin{eqnarray}
F^{(s)}&=&\vert \langle r^{(s)} \vert y^{(s)} \rangle \vert \nonumber \\
 &=& ^{\otimes n}\langle 0 \vert (\cos \phi^{(s)} I -i \sin \phi^{(s)} \Gamma_{\mu})(\cos \beta^{(s)} I -i \sin \beta^{(s)} \Gamma_{\mu}) \vert 0 \rangle^{\otimes n} \nonumber \\
 &=& \cos \beta^{(s)} ( \cos \phi^{(s)} + \sin \phi^{(s)}).
\end{eqnarray}
The learning rule is \cite{Shao}
\begin{eqnarray}
\theta_{j}(k+1)=\theta_{j}(k)+\eta\frac{\partial F^{(s)}(k)}{\partial \theta_{j}},
\end{eqnarray}
which results in
\begin{eqnarray}
\vert w(k+1) \rangle = \exp\left[i\sum_{j=0}^{2n} \left(\theta_{j}(k)+\eta\frac{\partial F^{(s)}(k)}{\partial \theta_{j}}\right)      \omega  \Gamma_{j}\right] \vert 0 \rangle^{\otimes n},
\end{eqnarray}
where $\eta$ is the learning-rate parameter of the back-propagation algorithm. \ \

We can determine $\langle \psi \vert \phi \rangle$ through the swap test. It consists in decomposition \cite{Walter}
\begin{eqnarray}
(C^{d})^{\otimes 2}=Sym^{2}(\mathbb{C}^{d})\oplus \Lambda^{2} (\mathbb{C}^{d}),
\end{eqnarray}
where $Sym^{2}(\mathbb{C}^{d})$ and $\Lambda^{2} (\mathbb{C}^{d})$ are symmetric and antisymmetric spaces, respectively. Consider the system in the state $\vert 0, \phi, \psi \rangle$. The Hadamard gate transform this state in
\begin{eqnarray}
\frac{1}{\sqrt{2}}(\vert 0, \phi, \psi \rangle + \vert 1, \phi, \psi \rangle).
\end{eqnarray}
Then, the controlled swap gate produces
\begin{eqnarray}
\frac{1}{\sqrt{2}}(\vert 0, \psi, \phi \rangle + \vert 1, \phi, \psi \rangle).
\end{eqnarray}
After a second application of Hadamard gate, we obtain
\begin{eqnarray}
\frac{1}{2} \vert 0 \rangle (\vert \psi, \phi \rangle + \vert \phi, \psi \rangle)+\frac{1}{2} \vert 1 \rangle (\vert \psi, \phi \rangle - \vert \phi, \psi \rangle).
\end{eqnarray}
If we perform a measurement on the first qubit, we get
\begin{eqnarray}
Pr(outcome =0)=\frac{1}{2}\left(1+\vert \langle \psi \vert \phi \rangle \vert^{2}\right).
\end{eqnarray}
By repeating $N$ times this procedure, we obtain
\begin{eqnarray}
\mid \langle \psi \vert \phi \rangle \mid=\left(1-2 \frac{\#\{outcome=0\}}{N}\right)^{1/2}.
\end{eqnarray}
\begin{definition}
Two Clifford quantum perceptrons $(\vert x' \rangle, \vert w' \rangle, \vert y' \rangle)$-type I and $(\vert x \rangle, \vert w \rangle, \vert y \rangle)$-Type II are called equivalent if $\vert y \rangle=\vert y' \rangle$.
\end{definition}

\begin{theorem}
Let $(\vert x' \rangle, \vert w' \rangle, \vert y' \rangle)$ and $(\vert x \rangle, \vert w \rangle, \vert y \rangle)$ be two Clifford quantum perceptrons of types I and II, respectively. Then every unitary transformation $U$ acting on CQP-Type II produces an equivalent CQP-Type I.
\end{theorem}
\begin{proof}
First, let $U$ be an unitary transformation such that
\begin{eqnarray}
U \vert x \rangle &=& \vert x' \rangle, \nonumber \\
U \vert w \rangle &=& \vert w' \rangle. \nonumber \\
\end{eqnarray}
We have that
\begin{eqnarray}
\phi'=\arccos \varphi(\langle x' \vert w' \rangle)=\arccos \varphi(\langle x \vert U^{\dag}U \vert w \rangle)=\arccos \varphi(\langle x \vert w \rangle)=\phi.
\end{eqnarray}
Since $\vert y \rangle$ is uniquely determined by $\phi$, $\vert y' \rangle=\vert y \rangle$. However we need to verify if the conditions of Definition 1 are satisfied. Notice that
\begin{eqnarray}
\vert x' \rangle &=&U \exp\left(i\sum_{j=0}^{2n-1} \alpha_{j} \Gamma_{j}\right) \vert 0 \rangle^{\otimes n} \nonumber \\
&=& UU'\vert 0 \rangle^{\otimes n} \nonumber \\
&=& U''\vert 0 \rangle^{\otimes n}, \nonumber \\
\end{eqnarray}
since $U'$ is unitary by Lemma 1 (if $A$ is Hermitian, $\exp(iA)$ is unitary. Again, by Lemma 1 $U''$ can be expressed by $U''=\exp\left(i\sum_{j=0}^{2^{n}-1} \omega_{j} \alpha''_{j} \Gamma_{j}^{(V)}\right)$. Analogously for the $\vert w' \rangle$ and $\vert y' \rangle$ and we provided a proof of theorem.
\end{proof}

\

Notice that the reverse is not true. \ \

There are several algorithms for Hamiltonian simulation. We can apply these techniques in our proposal, especially related to Definition 1. Particularly, we analyze the product formula approach \cite{Childs}. The exponential of a sum of operators is approximated by a product of exponentials. For the k-local Hamiltonians \cite{Nielsen} (a sum of $L$ Hermitian terms acting upon at most $k$ qubits), $L$ is upper bounded by a polynomial in $n$ and we get better first order bounds \cite{Childs}. So we have the following proposition

\begin{proposition}
Let M be a POVM related to measurement and $P_{U}$ and $P_{V}$ be the probabilities of obtaining the associated measurement outcome as operations
$U=exp\left(-it \sum_{j=1}^{2^{2n}} \omega_{j} \eta_{j}\Gamma_{j}^{(V)} \right)$ and $V=\left[\prod_{j=1}^{2^{2n}} \exp\left(\frac{-it}{r}  \omega_{j} \eta_{j}\Gamma_{j}^{(V)}\right)\right]^{r}$. Then $\vert P_{U}-P_{V} \vert \leq \frac{(2^{2n}t\max_{j} \vert  \eta_{j} \vert)^{2}}{r} $, where $r \in \mathbb{N}$ and $t \in \mathbb{R}$
\end{proposition}
\begin{proof}
The proof is an application of the results of references \cite{Nielsen} (Box 4.1, p. 195)) and \cite{Childs} (Proposition F.3). We have that
\begin{eqnarray}
\vert P_{U}-P_{V} \vert &\leq &  E(U,V) \nonumber \\
&\leq& \frac{(L \Lambda t)^{2}}{r}\exp \left(\frac{L \Lambda \vert t \vert}{r}\right)
\end{eqnarray}
where $E \equiv \max_{\vert \psi \rangle} \vert \vert (U- V) \vert \psi \rangle \vert \vert $. Let $H_{j}$ be an Hermitian operator \cite{Childs} and $\Lambda = \max_{j} \vert \vert H_{j} \vert \vert$. In this case $H_{j}=\eta_{j} \omega_{j} \Gamma_{j}^{(V)}$ and $\Gamma_{j}^{(V)}$ is unitary so that $\Lambda= \max_{j} \vert  \eta_{j} \vert$. Hence
\begin{eqnarray}
 \left\vert \left\vert \exp \left(-it \sum_{j=1}^{2^{2n}} \omega_{j} \eta_{j} \Gamma_{j}^{(V)}\right) - \left[\prod_{j=1}^{2^{2n}} \exp \left(\frac{-it}{r} \omega \eta_{j} \Gamma_{j}^{(V)}\right)  \right]^{r} \right\vert \right\vert \leq \frac{(2^{2n}t\max_{j} \vert  \eta_{j} \vert)^{2}}{r} \exp \left(\frac{2^{2n} \max_{j} \vert  \eta_{j} \vert \vert t \vert}{r}\right) \nonumber
\end{eqnarray}
\end{proof}

Notice that the choice $H_{j}=\omega \eta_{j}\Gamma_{j}^{V}$ is the worst scenario with analytic error bound $\frac{(2^{2n}t\max_{j} \vert  \eta_{j} \vert)^{2}}{r}$ for the fist-order product formula. The direct applying of definition of matrix multiplication give complexity $\Theta (2^{3n})$. However with $r=2n$, we obtain $\Theta (2n)$ for each operation in the CQP-Type II.

\begin{proposition}
The number of pairs of non-commuting elements for the basis of Clifford algebra $Cl(2n,0)$ is given by
\begin{eqnarray}
\Omega = \sum_{p,q; p < q}^{2n} \#_{p,q} \left({2n \choose p}{2n \choose q}\right),
\end{eqnarray}
where $ \#_{p,q}$ stands for a number of elements for which in the pairs $(p; q)$ ($p$-vector and $q$-vector spaces) we have, in the cases
\
(a)(odd; even) or (even; even) - an odd number of generators appearing simultaneously in the p-vector and q-vector;
\
(b)(odd; odd) - an even number of generators appearing simultaneously in the p-vector and q-vector.
\end{proposition}
\begin{proof}
The number of elements in a $p$-vector is given by ${number \ of \ generators \choose q}$ so that we must compute all possible combinations of $p$-vector and $q$-vector:
\begin{eqnarray}
&=&{2n \choose 1}{2n \choose 1}+{2n \choose 1}{2n \choose 2}+ \cdots {2n \choose 1}{2n \choose 2n} \nonumber \\
&+&{2n \choose 2}{2n \choose 1}+{2n \choose 2}{2n \choose 2}+ \cdots +{2n \choose 2}{2n \choose 2n} \nonumber \\
& \vdots & \nonumber \\
&+&{2n \choose 2n}{2n \choose 1}+{2n \choose 2n}{2n \choose 2}+ \cdots +{2n \choose 2n}{2n \choose 2n} \nonumber \\
\end{eqnarray}
Most of these terms appear repeatedly so that we should only count them once. In order for the elements of the basis do not commute, we need an odd number of anticommutations of the generators in the product of a p-vector and a q-vector, $\gamma_{i_{1}}\gamma_{i_{2}},...\gamma_{i_{p}}\gamma_{j_{1}}\gamma_{j_{2}},...\gamma_{j_{q}}$, i.e
\begin{eqnarray}
\gamma_{i_{k}}\gamma_{j_{1}}\gamma_{j_{2}},...\gamma_{j_{q}}=-\gamma_{j_{1}}\gamma_{j_{2}},...\gamma_{j_{q}}\gamma_{i_{k}}
\end{eqnarray}
 an odd number of times. This occurs if we have an odd number of generators of the p-vector (which do not appear simultaneously in the q-vector) and an odd number of generators of the q-vector. Consider the case (even; odd) and suppose initially that we have an even number of generators that appear simultaneously in both the p-vector and q-vector. Consequently, we have an even number of generators (which appear simultaneously in the multivectors) that anticommute and another even number of vectors (which do not appear simultaneously in the multivectors) that anticommute with an odd number of generators of the q-vector. So we have an even total number of anticommutations and therefore the base elements commute. Consider now an odd number of generators that appear simultaneously in both the p-vector and q-vector. So we have an odd number of generators (which appear simultaneously in the multivectors) that anticommute with an even number of generators and another odd number of generators that anticommute with an odd number of generators. Consequently, we have an odd total number of anticommutations so that the basis elements anticommute.
The development is similar for the other cases.
\end{proof}

A corollary of this proposition is the commutator bounds of errors. An application of bound for the first-order formula obtained by \cite{Childs} (Theorem F.5)
\begin{eqnarray}
\left\vert \left\vert \exp \left(-it \sum_{j=1}^{2^{2n}} \omega_{j} \eta_{j} \Gamma_{j}^{(V)}\right) - \left[\prod_{j=1}^{2^{2n}} \exp \left(-it \omega_{j} \eta_{j} \Gamma_{j}^{(V)}\right)  \right]^{r} \right\vert \right\vert \nonumber \\
\leq \Omega \frac{(\max_{j} \vert \eta_{j}  \vert t)^{2}}{r}+ \frac{(2^{2n} \vert t \vert^{3}\max_{j} \vert  \eta_{j} \vert)^{2}}{3r^{2}} exp \left(\frac{2^{2n} \max_{j} \vert  \eta_{j} \vert \vert t \vert}{r}\right) \nonumber
\end{eqnarray}
for the CQP-Type I. In the case where we are considering only the algebra generators as for the CQP-Type II, we have
\begin{eqnarray}
\left\vert \left\vert \exp \left(-it \sum_{j=1}^{2n}  \eta_{j} \Gamma_{j}\right) - \left[\prod_{j=1}^{2n} \exp \left(-it  \eta_{j} \Gamma_{j}\right)  \right]^{r} \right\vert \right\vert \nonumber \\
\leq 2n \frac{(\max_{j} \vert \eta_{j}  \vert t)^{2}}{r}+ \frac{(2n \vert t \vert^{3}\max_{j} \vert  \eta_{j} \vert)^{2}}{3r^{2}} exp \left(\frac{2n \max_{j} \vert  \eta_{j} \vert \vert t \vert}{r}\right) \nonumber
\end{eqnarray}

A second proposal capable of allowing the implementation of an arbitrary activation function through functions of operators is given by following definition

\begin{definition}
A Clifford quantum neural network is given by unitary operators $U^{Cliff}(x), U^{Cliff}(\theta)$, an state $\vert 0 \rangle ^{\otimes n}$, a Hermitian operator $A_{k}$, a real activation function $\phi$ and a output $\langle x; \theta \vert \phi(A_{k}) \vert x; \theta \rangle$  defined as
\begin{eqnarray}
\vert x \rangle = U^{Cliff}(\theta)U^{Cliff}(x)\vert 0 \rangle^{\otimes n} = \sum_{i=1}^{n} a_{i; \theta} \vert x_{i} \rangle,
\end{eqnarray}
\begin{eqnarray}
A_{k}= \sum_{i} \vert a_{i; \theta} \vert \vert x_{i} \rangle \langle  x_{i} \vert
\end{eqnarray}
and
\begin{eqnarray}
\phi (A_{k})= \sum_{j}\phi (\vert a_{j; \theta} \vert ) \vert x_{j}^{k} \rangle \langle  x_{j}^{k} \vert
\end{eqnarray}
with $\sum_{i} \vert a_{i; \theta} \vert^{2}=1$.
\end{definition}
Notice that $\phi(A_{k})$ is also a Hermitian operator. The superscript index in the unitary operators means that they are constructed through elements of Clifford algebra, which can reflect the geometric properties of the system.
Another possibility is given by
\begin{definition}
A Clifford quantum neural network is given by unitary operators $U^{Cliff}(x), U^{Cliff}(y), U^{Cliff}(\theta)$, the states $\vert 0 \rangle ^{\otimes n}$, $\vert y_{out} \rangle$ and $\vert y_{out,d} \rangle$ a Hermitian operator $A= \sum_{j} \vert a_{j}(x_{j}; \theta)\vert \vert x_{j} \rangle \langle  x_{j} \vert $, a real activation function $\phi$ and a output defined as
\begin{eqnarray}
\vert x \rangle = U^{Cliff}(\theta)U^{Cliff}(x)\vert 0 \rangle^{\otimes n} = \sum_{i=1}^{n} a_{i; \theta} \vert x_{i} \rangle,
\end{eqnarray}
\begin{eqnarray}
 A= \sum_{j} \vert a_{j}(x_{j}; \theta) \vert \vert x_{j} \rangle \langle  x_{j} \vert,
\end{eqnarray}
\begin{eqnarray}
\phi (A)= \sum_{j}\phi (\vert a_{j; \theta} \vert) \vert x_{j} \rangle \langle  x_{j} \vert
\end{eqnarray}
\begin{eqnarray}
\vert y_{out} \rangle = \phi(A) \vert x_{i} \rangle =\frac{1}{N} \sum_{j}\phi (\vert a_{j; \theta} \vert ) \vert x_{j} \rangle \langle  x_{j} \vert x_{i} \rangle,
\end{eqnarray}
\begin{eqnarray}
\vert y_{out,d} \rangle = U^{Cliff}(y) \vert 0 \rangle ^{\otimes n},
\end{eqnarray}
where  $\sum_{i} \vert a_{i; \theta} \vert^{2}=1$ $N= \sqrt{\langle x_{i} \vert f^{\dag}f(A) \vert x_{i}}$.
\end{definition}
In this scheme, we may use the quantum fidelity $F=\langle y_{out}(\theta) \vert y_{out,d} \rangle$ as cost function. \\

All these proposals for quantum neural networks based on clifford algebras may be trained by a classical optimization algorithms as in several other proposals for hybrid algorithms \cite{Maria}. An example of a circuit for implementing unitary operations is given in Appendix B.  \ \

Alternatively, we can build the Hermitian operators from a basis of space $[\bigwedge^{0} \mathbb{R}^{3,0} \oplus \bigwedge^{1}\mathbb{R}^{3,0}]^{\otimes n}$ that corresponds to subspace of tensor product of Clifford algebras $Cl(3,0)^{\otimes n}$.
\begin{proposition}
The subspace $[\bigwedge^{0} \mathbb{R}^{3,0} \oplus \bigwedge^{1}\mathbb{R}^{3,0}]^{\otimes n}  \subset Cl(3,0)^{\otimes n} $ is isomorphic to the space vector of Hermitian matrices $2^{n} \times 2^{n}$.
\end{proposition}
\begin{proof}
A basis of subspace $[\bigwedge^{0} \mathbb{R}^{3,0} \oplus \bigwedge^{1}\mathbb{R}^{3,0}]^{\otimes n}  \subset Cl(3,0)^{\otimes n}$ is given by $\{1, \gamma_{1}, \gamma_{2}, \gamma_{3}\}^{\otimes n}$, for which the elements have a representation in terms of matrices $\{I, \sigma_{x}, \sigma_{y}, \sigma_{z}\}^{\otimes n}$. Given a vector space $V$ of dimension $2^{2n}$, every linearly independent subset of $2^{2n}$ elements is a basis of $V$. The $2^{2n}$ elements are linearly independent and Hermitian matrices. Therefore they are form a basis of space of Hermitian matrices $2^{n} \times 2^{n}$.
\end{proof}

We will show in the next section that a formulation based on the tensor product of  Clifford algebras $Cl(3,0)^{\otimes n}$ is associated with a proposal of generalization of the quantum Fourier transform, which can be used for our proposal of quantum neural networks.

\section{Generalized quantum fourier transform}

Unitary operators for the quantum neural networks based on Clifford algebras can also be obtained from a generalization of quantum fourier transform given by
\begin{eqnarray}
\vert j \rangle \mapsto \frac{1}{2^{n/2}} \sum_{k=0}^{2^{n}-1}e^{2\pi ij\frac{k}{2^n}+i \theta \Gamma_{k}} \vert k \rangle,
\end{eqnarray}
where
\begin{eqnarray}
\Gamma_{k}=(\alpha_{x_{1}}^{k_{1}}\sigma_{x}+\alpha_{y_{1}}^{k_{1}}\sigma_{y}+\alpha_{z_{1}}^{k_{1}}\sigma_{z})\otimes  \cdots \otimes I+ \cdots +I \otimes \cdots  \otimes (\alpha_{x_{n}}^{k_{n}}\sigma_{x}+\alpha_{y_{n}}^{k_{n}}\sigma_{y}+\alpha_{z_{n}}^{k_{n}}\sigma_{z}),
\end{eqnarray}
with $(\alpha_{x_{i}}, \alpha_{y_{i}}, \alpha_{z_{i}})$ is a unit vector. Each $k$ corresponds to a set $(k_{1},\cdots,k_{i}, \cdots,k_{n})$, $k_{i}=0,1$. $k_{i}$ is not an exponent, it is an index.
Then
\begin{eqnarray}
\vert j \rangle &\mapsto&\frac{1}{2^{n/2}} \sum_{k_{1}=0}^{1} \cdots \sum_{k_{n}=0}^{1}e^{2\pi ij(\sum_{l=1}^{n}k_{l}2^{-l})+i \theta [ (\alpha_{x_{1}}^{k_{1}}\sigma_{x}+\alpha_{y_{1}}^{k_{1}}\sigma_{y}+\alpha_{z_{1}}^{k_{1}}\sigma_{z})\otimes  \cdots \otimes I+ \cdots +I \otimes \cdots  \otimes (\alpha_{x_{n}}^{k_{n}}\sigma_{x}+\alpha_{y_{n}}^{k_{n}}\sigma_{y}+\alpha_{z_{n}}^{k_{n}}\sigma_{z})] } \nonumber \\
&\times& \vert k_{1} \cdots k_{n} \rangle. \nonumber
\end{eqnarray}
This expression can be factored
\begin{eqnarray}
 \frac{1}{2^{n/2}} \sum_{k_{1}=0}^{1} \cdots \sum_{k_{n}=0}^{1}e^{2\pi ij(\sum_{l=1}^{n}k_{l}2^{-l})}e^{i \theta[ (\alpha_{x_{1}}^{k_{1}}\sigma_{x}+\alpha_{y_{1}}^{k_{1}}\sigma_{y}+\alpha_{z_{1}}) \otimes \cdots
\otimes I]} \cdots e^{i \theta[I \otimes... \otimes (\alpha_{x_{n}}^{k_{n}}\sigma_{x}+\alpha_{y_{n}}^{k_{n}}\sigma_{y}+\alpha_{z_{n}}^{k_{n}}\sigma_{z})]} \vert k_{1} \cdots k_{n} \rangle \nonumber
\end{eqnarray}

\begin{eqnarray}
&=&\frac{1}{2^{n/2}} \sum_{k_{1}=0}^{1} \cdots \sum_{k_{n}=0}^{1}e^{2\pi ij(\sum_{l=1}^{n}k_{l}2^{-l})} [e^{i \theta (\alpha_{x_{1}}^{k_{1}}\sigma_{x}+\alpha_{y_{1}}^{k_{1}}\sigma_{y}+\alpha_{z_{1}}^{k_{1}}\sigma_{z})} \otimes I \cdots \otimes I] \cdots \nonumber \\
&\times &[I \otimes  \cdots \otimes e^{i\theta (\alpha_{x_{n}}^{k_{n}}\sigma_{x}+\alpha_{y_{n}}^{k_{n}}\sigma_{y}+\alpha_{z_{n}}^{k_{n}}\sigma_{z})}] \vert k_{1} \cdots k_{n} \rangle \nonumber \\
&=&\frac{1}{2^{n/2}}\sum_{k_{1}=0}^{1} \cdots \sum_{k_{n}=0}^{1}e^{2 \pi j k_{1}2^{-1}}c^{i \theta (\alpha_{x_{1}}^{k_{1}}\sigma_{x}+\alpha_{y_{1}}^{k_{1}}\sigma_{y}+\alpha_{z_{1}}^{k_{1}}\sigma_{z})} \otimes \cdots \nonumber \\
&\otimes & e^{2 \pi j k_{n}2^{-n}}c^{i \theta (\alpha_{x_{n}}^{k_{n}}\sigma_{x}+\alpha_{y_{n}}^{k_{n}}\sigma_{y}+\alpha_{z_{n}}^{k_{n}}\sigma_{z})} \vert k_{1} \cdots k_{n} \rangle
\end{eqnarray}

\begin{eqnarray}
&=&\frac{1}{2^{n/2}} \bigotimes_{l=0}^{n} \left[\sum_{k_{l}=0}^{1}e^{2 \pi ijk_{l}2^{-l}}e^{i \theta(\alpha_{x_{l}}^{k_{l}}\sigma_{x}+\alpha_{y_{l}}^{k_{l}}\sigma_{y}+\alpha_{z_{l}}^{k_{1}}\sigma_{z})}\right] \vert k_{l} \rangle \nonumber \\
&=&\frac{1}{2^{n/2}} \bigotimes_{l=0}^{n} \left[e^{i \theta (\alpha_{x_{l}}^{0}\sigma_{x}+\alpha_{y_{l}}^{0}\sigma_{y}+\alpha_{z_{l}}^{0}\sigma_{z})} \vert 0 \rangle + e^{2 \pi ij2^{-l}}e^{i \theta(\alpha_{x_{l}}^{1}\sigma_{x}+\alpha_{y_{l}}^{1}\sigma_{y}+\alpha_{z_{l}}^{1}\sigma_{z})} \vert 1 \rangle \right] \nonumber
\end{eqnarray}

\begin{eqnarray}
&=&\frac{1}{2^{n/2}} \left[(e^{i \theta (\alpha_{x_{l}}^{0}\sigma_{x}+\alpha_{y_{l}}^{0}\sigma_{y}+\alpha_{z_{l}}^{0}\sigma_{z})}) \vert 0 \rangle
+ e^{2 \pi ij 2^{-1}} (e^{i \theta (\alpha_{x_{l}}^{1}\sigma_{x}+\alpha_{y_{l}}^{1}\sigma_{y}+\alpha_{z_{l}}^{1}\sigma_{z})}) \vert 1 \rangle \right] \nonumber \\
&\otimes& \cdots \left[(e^{i \theta (\alpha_{x_{n}}^{0}\sigma_{x}+\alpha_{y_{n}}^{0}\sigma_{y}+\alpha_{z_{n}}^{0}\sigma_{z})}) \vert 0 \rangle
+ e^{2 \pi ij 2^{-1}} (e^{i \theta (\alpha_{x_{n}}^{1}\sigma_{x}+\alpha_{y_{n}}^{1}\sigma_{y}+\alpha_{z_{n}}^{1}\sigma_{z})}) \vert 1 \rangle \right] \nonumber
\end{eqnarray}

\begin{eqnarray}
&=& \frac{1}{2^{n/2}} \left[(e^{i \theta (\alpha_{x_{l}}^{0}\sigma_{x}+\alpha_{y_{l}}^{0}\sigma_{y}+\alpha_{z_{l}}^{0}\sigma_{z})}) \vert 0 \rangle
+ e^{2 \pi i 0\cdot j_{n}} (e^{i \theta (\alpha_{x_{l}}^{1}\sigma_{x}+\alpha_{y_{l}}^{1}\sigma_{y}+\alpha_{z_{l}}^{1}\sigma_{z})}) \vert 1 \rangle \right] \nonumber \\
&\otimes&  \left[(e^{i \theta (\alpha_{x_{2}}^{0}\sigma_{x}+\alpha_{y_{2}}^{0}\sigma_{y}+\alpha_{z_{2}}^{0}\sigma_{z})}) \vert 0 \rangle
+ e^{2 \pi ij 0 \cdot j_{n-1}j_{n}} (e^{i \theta (\alpha_{x_{2}}^{1}\sigma_{x}+\alpha_{y_{2}}^{1}\sigma_{y}+\alpha_{z_{2}}^{1}\sigma_{z})}) \vert 1 \rangle \right] \nonumber \\
&\cdots& \otimes \left[(e^{i \theta (\alpha_{x_{n}}^{0}\sigma_{x}+\alpha_{y_{n}}^{0}\sigma_{y}+\alpha_{z_{n}}^{0}\sigma_{z})}) \vert 0 \rangle
+ e^{2 \pi i 0 \cdot j_{1}j_{2} \cdots j_{n}} (e^{i \theta (\alpha_{x_{n}}^{1}\sigma_{x}+\alpha_{y_{n}}^{1}\sigma_{y}+\alpha_{z_{n}}^{1}\sigma_{z})}) \vert 1 \rangle \right]. \nonumber
\end{eqnarray}

Now we will prove the unitarity of the transform. For that, consider the following Lemma

\begin{lemma}
Let $R_{\widehat{n}(\theta)}$ be a unitary rotation operator and $\vert k \rangle$ a computational basis states for a qubit. Then $\sum_{k=0}^{1}R_{\widehat{n}(\theta)}\vert k \rangle \langle k \vert R_{\widehat{n}(\theta)}^{\dag}=I$.
\end{lemma}

\begin{proof}
The unitary rotation operator can be expressed as \cite{Nielsen}
\begin{eqnarray}
R_{\widehat{n}(\theta)}=exp(-i \theta \widehat{n} \cdot \frac{\overrightarrow{\sigma}}{2})=\cos(\frac{\theta}{2})I-i\sin(\frac{\theta}{2})(n_{x}\sigma_{x}+n_{y}\sigma_{y}+n_{z}\sigma_{z}).
\end{eqnarray}
Therefore
\begin{eqnarray}
\sum_{k=0}^{1}R_{\widehat{n}(\theta)}\vert k \rangle \langle k \vert R_{\widehat{n}(\theta)}^{\dag}&=&\sum_{k=0}^{1}\cos(\frac{\theta}{2})I-i\sin(\frac{\theta}{2})(n_{x}(\vert 0 \rangle \langle 1 \vert + \vert 1 \rangle \langle 0 \vert)+n_{y}(-i\vert 0 \rangle \langle 1 \vert +i\vert 1 \rangle \langle 0 \vert) \nonumber \\
&=&+n(z) (\vert 0 \rangle 0 \vert - \vert 1 \rangle 1 \vert)) \nonumber \\
&=& (\cos^{2} (\theta /2)+n_{z}^{2} \sin^{2} (\theta /2)+[-n_{y}\sin(\theta /2)-i n_{x}\sin(\theta /2)] \nonumber \\
&\times&[-n_{y}\sin(\theta /2)+i n_{x}\sin(\theta /2)])\vert 0 \rangle \langle 0 \vert \nonumber \\
&+& (n_{y}^{2} \sin^{2}(\theta/2)+n_{x}^{2} \sin^{2}(\theta/2) \nonumber \\
&+& [\cos(\theta/2)+i \sin (\theta /2)n_{z})] [\cos(\theta/2)-i \sin (\theta /2)n_{z}])\vert 1 \rangle \langle 1 \vert \nonumber \\
&+& [n_{y} \cos(\theta/2)\sin(\theta/2)+in_{x}\cos(\theta/2)\sin(\theta/2)-in_{z}n_{y}\sin^{2}(\theta/2) \nonumber \\
&+&n_{x}n_{z}\sin^{2}(\theta/2)
+[-n_{y}\sin (\theta /2)-in_{x}\sin (\theta /2)] \nonumber \\
&+& [\cos(\theta/2)-i  \sin (\theta /2)]] \vert 0 \rangle \langle 1 \vert \nonumber \\
&+& [n_{y} \cos(\theta/2)\sin(\theta/2)+in_{y}^{2}\sin^{2}(\theta/2)-in_{x}\sin^{2}(\theta/2)\cos(\theta/2)\nonumber \\
&+&n_{x}n_{y}\sin^{2}(\theta/2)
+[-n_{y}\sin (\theta /2)+in_{x}\sin (\theta /2)] \nonumber \\
&+& [\cos(\theta/2)+in_{z} \sin (\theta /2)]] \vert 1 \rangle \langle 0 \vert \nonumber \\
&=&  \vert 0 \rangle \langle 0 \vert + \vert 1 \rangle \langle 1 \vert =I.
\end{eqnarray}
\end{proof}

\begin{theorem}
The generalized quantum Fourier transform $F_{N}=\frac{1}{\sqrt{N}} \sum_{j,k=0}^{N-1} e^{\frac{2 \pi i j k}{N}} e^{i \theta \Gamma_{k}} \vert k \rangle \langle j \vert$ is unitary.
\end{theorem}

\begin{proof}
We have that $F_{N}^{\dag}$ is given by
\begin{eqnarray}
F_{N}^{\dag}=\frac{1}{\sqrt{N}} \sum_{j,k=0}^{N-1} e^{-\frac{2 \pi i j' k'}{N}}  \vert j' \rangle \langle k' \vert e^{-i \theta \Gamma_{k'}}
\end{eqnarray}
Consequently

\begin{eqnarray}
F_{N}F_{N}^{\dag}&=&\frac{1}{N} \sum_{j,k=0}^{N-1} \sum_{j',k'=0}^{N-1} e^{2 \pi i \frac{(jk-j'k')}{N}} e^{i \theta \Gamma_{k}} \vert k \rangle \langle j \vert \vert j' \rangle \langle k' \vert e^{-i \theta \Gamma_{k'}} \nonumber \\
&=& \frac{1}{N} \sum_{j,k=0}^{N-1} \sum_{j',k'=0}^{N-1} e^{2 \pi i \frac{(jk-j'k')}{N}} \langle j \vert j' \rangle e^{i \theta \Gamma_{k}} \vert k \rangle \langle k' \vert e^{-i \theta_{k'} \Gamma_{k'}} \nonumber \\
&=&  \frac{1}{N} \sum_{k,k'}N \delta_{kk'} e^{i \theta} \Gamma_{k} \vert k \rangle \langle k' \vert e^{-i \theta_{k'} \Gamma_{k'}} \nonumber \\
&=& \sum_{k} e^{i \theta \Gamma_{k}} \vert k \rangle \langle k \vert e^{-i \theta \Gamma_{k}},
\end{eqnarray}
where we use that $\sum_{j} e^{\frac{2 \pi i(k-k')j}{N}}= N \delta_{k,k'}$. Therefore
\begin{eqnarray}
&=&\sum_{k_{1},...,k_{n}}e^{i \theta(\alpha_{x_{l}}^{k_{1}}\sigma_{x}+\alpha_{y_{l}}^{k_{1}}\sigma_{y}+\alpha_{z_{l}}^{k_{1}}\sigma_{z}) \otimes \cdots I+ \cdots +I \otimes...\otimes (\alpha_{x_{n}}^{k_{n}}\sigma_{x}+\alpha_{y_{n}}^{k_{n}}\sigma_{y}+\alpha_{z_{n}}^{k_{n}}\sigma_{z})} \nonumber \\
&\times& \vert k_{1} \cdots k_{n} \rangle \langle  k_{1} \cdots k_{n} \vert \nonumber \\
&\times& \sum_{k_{1},...,k_{n}}e^{i \theta(\alpha_{x_{l}}^{k_{1}}\sigma_{x}+\alpha_{y_{l}}^{k_{1}}\sigma_{y}+\alpha_{z_{l}}^{k_{1}}\sigma_{z}) \otimes \cdots I+ \cdots +I \otimes...\otimes (\alpha_{x_{n}}^{k_{n}}\sigma_{x}+\alpha_{y_{n}}^{k_{n}}\sigma_{y}+\alpha_{z_{n}}^{k_{n}}\sigma_{z})} \nonumber
\end{eqnarray}

\begin{eqnarray}
&=&\sum_{k_{1}} e^{i \theta(\alpha_{x_{l}}^{k_{1}}\sigma_{x}+\alpha_{y_{l}}^{k_{1}}\sigma_{y}+\alpha_{z_{l}}^{k_{1}}\sigma_{z})}\vert k_{1} \rangle \langle k_{1} \vert e^{-i \theta(\alpha_{x_{l}}^{k_{1}}\sigma_{x}+\alpha_{y_{l}}^{k_{1}}\sigma_{y}+\alpha_{z_{l}}^{k_{1}}\sigma_{z})} \nonumber \\
&\otimes& \cdots \otimes \nonumber \\
&\otimes& \sum_{k_{n}} e^{i \theta(\alpha_{x_{n}}^{k_{n}}\sigma_{x}+\alpha_{y_{n}}^{k_{n}}\sigma_{y}+\alpha_{z_{n}}^{k_{n}}\sigma_{z})}\vert k_{n} \rangle \langle k_{n} \vert e^{-i \theta(\alpha_{x_{n}}^{k_{n}}\sigma_{x}+\alpha_{y_{n}}^{k_{n}}\sigma_{y}+\alpha_{z_{n}}^{k_{n}}\sigma_{z})}. \nonumber
\end{eqnarray}
Using the previous Lemma
\begin{eqnarray}
F_{N}F_{N}^{\dag}&=&(\vert 0 \rangle \langle 0 \vert + \vert 1 \rangle \langle 1 \vert) \otimes \cdots \otimes (\vert 0 \rangle \langle 0 \vert + \vert 1 \rangle \langle 1 \vert) \nonumber \\
&=& I^{\otimes n}.
\end{eqnarray}
It is easy to verify that $F_{N}^{\dag}F_{N}=I^{\otimes n}$ so that the proof is finished.
\end{proof}

We can see that generalized quantum Fourier transform is implemented through group
\begin{eqnarray}
\bigotimes_{i=1}^{n}  spin_{+}(3,0) =  spin_{+}(3,0) \otimes \cdots \otimes spin_{+}(3,0).
\end{eqnarray}
or equivalently
\begin{eqnarray}
\bigoplus_{i=1}^{n} Lie [spin_{+}(3,0)]=Lie [spin_{+}(3,0)]\oplus \cdots Lie [spin_{+}(3,0)].
\end{eqnarray}
In fact, consider the following theorem
\begin{theorem}
Let $\rho$ and $\rho^{'}$ be two representations of $\left([\bigwedge^{2} \mathbb{R}^{p,q}]^{\oplus^{n}}, [\ , \ ]\right)$ and $U \left( Lie [spin_{+}(p,q)]\right)$, respectively, where $U( \  )$ is the universal enveloping algebra defined by
\begin{eqnarray}
U \left( Lie [spin_{+}(p,q)]\right)=\frac{T[U \left( Lie [spin_{+}(p,q)]\right)]}{I_{d}},
\end{eqnarray}
where $T$ is the tensor algebra an $I_{d}$ is the two-sided ideal generated by elements of the form $x \otimes y - y \otimes x -[x,y]$. Then
\end{theorem}
\begin{eqnarray}
\rho'(Lie [spin_{+}(p,q))\otimes \cdots I+ \cdots + I \otimes \cdots \rho'(Lie [spin_{+}(p,q)) \simeq \rho \left([\bigwedge^{2} \mathbb{R}^{p,q}]^{\oplus^{n}}, [\ , \ ]\right).
\end{eqnarray}
\begin{proof}
Initially, we prove that $\left([\bigwedge^{2} \mathbb{R}^{p,q}]^{\oplus^{n}}, [\ , \ ]\right)$ is the direct sum  $\bigoplus_{i=1}^{n} Lie [spin_{+}(p,q)]$. So we first consider $[B,B']=([B_{1},B'_{1}],...[B_{n},B'_{n}]) \in [\bigwedge^{2} \mathbb{R}^{p,q}]^{\oplus^{n}}$ with $[B_{i},B_{j}]\equiv B_{i}B_{j}-B_{j}B_{i}$. Then
\begin{eqnarray}
BB'&=&(\langle B_{1}B'_{1} \rangle_{0}+ \langle B'_{1}B_{1} \rangle_{2}+\langle B_{1}B'_{1} \rangle_{4}, \cdots, \langle B_{n}B'_{n} \rangle_{0}+ \langle B'_{n}B_{n} \rangle_{2}+\langle B_{n}B'_{n} \rangle_{4}) \nonumber \\
&\in& \left[\mathbb{R} \oplus  \bigwedge^{2} \mathbb{R}^{p,q} \oplus \bigwedge^{4} \mathbb{R}^{p,q}\right]^{\oplus n}.
\end{eqnarray}
We have that
\begin{eqnarray}
(BB')^{\sim}&=&((B_{1}B'_{1})^{\sim}, \cdots (B_{n}B'_{n})^{\sim}) \nonumber \\
&=&(\langle B_{1}B'_{1} \rangle_{0}- \langle B'_{1}B_{1} \rangle_{2}+\langle B_{1}B'_{1} \rangle_{4}, \cdots, \langle B_{n}B'_{n} \rangle_{0}- \langle B'_{n}B_{n} \rangle_{2}+\langle B_{n}B'_{n} \rangle_{4}). \nonumber
\end{eqnarray}
On the other hand
\begin{eqnarray}
(BB')^{\sim}&=&((B_{1}B'_{1})^{\sim}, \cdots,(B_{1}B'_{1})^{\sim}) \nonumber \\
&=&(\widetilde{B'}_{1}\widetilde{B}_{1}, \cdots, \widetilde{B'}_{n}\widetilde{B}_{n}) \nonumber \\
&=&(B'_{1}B_{1},...,B'_{n}B_{n}).
\end{eqnarray}
Consequently,
\begin{eqnarray}
[B,B']=(2 \langle B_{1}B'_{1} \rangle_{2}, \cdots, 2 \langle B_{n}B'_{n} \rangle_{2}),
\end{eqnarray}
i.e., we show that
\begin{eqnarray}
\left([\bigwedge^{2} \mathbb{R}^{p,q}], [\ , \ ]\right)\oplus \cdots \oplus \left([\bigwedge^{2} \mathbb{R}^{p,q}], [\ , \ ]\right)= \bigoplus_{i=1}^{n} Lie [spin_{+}(p,q)].
\end{eqnarray}
Now we define the map $\psi$ given by:
\begin{eqnarray}
\psi(\rho (B_{1},...B_{n}))=\rho'(B_{1}) \otimes \cdots \otimes I+ \cdots I \otimes  \cdots \otimes \rho'(B_{n}).
\end{eqnarray}
Thus
\begin{eqnarray}
\psi([\rho(B_{1},..., B_{n}),\rho(C_{1},..., C_{n})]&=&\psi(\rho([B_{1},C_{1}],...[B_{n},C_{n}]) \nonumber \\
&=&\rho'([B_{1},C_{1}]) \otimes \cdots I \otimes \cdots \rho'([B_{n},C_{n}]) \nonumber \\
&=& \rho'(B{_1}C_{1} -C_{1}B_{1}) \cdots \otimes I+ \cdots I \otimes \cdots \rho'(B_{n}C_{n}-C_{n}B_{n}). \nonumber
\end{eqnarray}
since we are considered universal enveloping algebra. Therefore
\begin{eqnarray}
\psi([\rho(B_{1},..., B_{n}),\rho(C_{1},..., C_{n})]&=& \rho'(B_{1})\rho'(C_{1}) \cdots \otimes I- \rho'(C_{1})\rho'(B_{1}) \cdots \otimes I \nonumber \\
&+& \cdots \nonumber \\
&+& I \otimes \cdots \rho'(B_{n})\rho'(C_{n})- I \otimes \cdots \rho'(C_{n})\rho'(B_{n}) \nonumber \\
&=& [\rho'(B_{1}) \otimes \cdots \otimes I+ \cdots + I \otimes \cdots \otimes \rho'(B_{n})] \nonumber \\
&\times& [\rho'(C_{1}) \otimes \cdots \otimes I+ \cdots + I \otimes \cdots \otimes \rho'(C_{n})] \nonumber \\
&-& [\rho'(C_{1}) \otimes \cdots \otimes I+ \cdots + I \otimes \cdots \otimes \rho'(C_{n})] \nonumber \\
& \times&[\rho'(B_{1}) \otimes \cdots \otimes I+ \cdots + I \otimes \cdots \otimes \rho'(B_{n})] \nonumber \\
&=& [\rho'(B_{1}) \otimes \cdots \otimes I+ \cdots \otimes I \cdots \otimes \rho'(B_{n}) \nonumber \\
&,& \rho'(C_{1}) \otimes \cdots \otimes I+ \cdots \otimes I \cdots \otimes \rho'(c_{n})] \nonumber \\
&=&[\psi(\rho(B_{1},...B_{n})),\psi(\rho(B_{1},...B_{n}))]. \nonumber
\end{eqnarray}
Then we have a bijective homomorphism and the proof is finished.
\end{proof}

Thus $(e^{B_{k}})^{\otimes n} \in \bigoplus spin_{+}(p,q)$ with $B \in \bigwedge \mathbb{R}^{p,q}$. This result is general and we can particularize to $\mathbb{R}^{3,0}$. The next theorem gives us the distance between the generalized quantum Fourier transform and the quantum Fourier transform

\begin{theorem}
Let $F_{G}$ be generalized quantum Fourier transform and $F$ be quantum Fourier transform. Then
\begin{eqnarray}
 \| F_{G}-F \|\leq 2^{\frac{3\log N}{2}}\theta \log N\sqrt{2}e^{\theta \log N \sqrt{2}},
 \end{eqnarray}
 where $\parallel A  \parallel=  \sqrt{\sum_{i,j=1}^{n}\vert a_{ij} \vert^{2}}$ \cite{Horn1}.
\end{theorem}

\begin{proof}
We have that
\begin{eqnarray}
\| F_{G}-F \|&=& \left\|\frac{1}{\sqrt{N}} \sum_{j,k=0}^{N-1}e^{\frac{2 \pi ijk}{N}+ i \theta \Gamma_{k}} \vert k \rangle \langle j \vert -
\frac{1}{\sqrt{N}} \sum_{j,k=0}^{N-1}e^{\frac{2 \pi ijk}{N}} \vert k \rangle \langle j \vert\right\| \nonumber \\
&=& \left\|\frac{1}{\sqrt{N}} \sum_{j,k=0}^{N-1}(e^{\frac{2 \pi ijk}{N}+ i \theta \Gamma_{k}}-e^{\frac{2 \pi ijk}{N}}) \vert k \rangle \langle j \vert \right\| \nonumber \\
&\leq& \frac{1}{\sqrt{N}} \sum_{j,k=0}^{N-1}\left\|(e^{\frac{2 \pi ijk}{N}+ i \theta \Gamma_{k}}-e^{\frac{2 \pi ijk}{N}}) \vert k \rangle \langle j \vert \right\| \nonumber \\
&\leq&  \frac{1}{\sqrt{N}} \sum_{j,k=0}^{N-1} \left\| e^{\frac{2 \pi ijk}{N}+ i \theta \Gamma_{k}}-e^{\frac{2 \pi ijk}{N}} I^{\otimes N} \right\| \| \vert k \rangle \langle j \vert \| \nonumber \\
&=& \frac{1}{\sqrt{N}} \sum_{j,k=0}^{N-1} \left\| e^{\frac{2 \pi ijk}{N}+ i \theta \Gamma_{k}}-e^{\frac{2 \pi ijk}{N}} I^{\otimes N} \right\|. \nonumber \\
\end{eqnarray}
Using $\|e^{X+Y}-e^{X}\| \leq \|Y\|e^{\|X\|}e^{\|Y\|}$ \cite{Horn2},
\begin{eqnarray}
\| F_{G}-F \| \leq \frac{1}{\sqrt{N}} \sum_{j,k=0}^{N-1} \|i \theta \Gamma_{k} \| e^{\| 2 \pi ijk/N I^{\otimes N}\|} e^{\|i \theta \Gamma_{k} \|}.
\end{eqnarray}
Since $\|X \otimes Y \| = \| X \| \| Y \|$,
\begin{eqnarray}
\| F_{G}-F \| &\leq& \frac{1}{\sqrt{N}} \sum_{j=0}^{N-1}\sum_{k_{1}=0}^{1} \cdots \sum_{k_{n}=0}^{1} \|i \theta [(\alpha_{x_{l}}^{k_{1}}\sigma_{x}+\alpha_{y_{l}}^{k_{1}}\sigma_{y}+\alpha_{z_{l}}^{k_{1}}\sigma_{z}) \otimes \cdots \otimes I+ \cdots \nonumber \\
&+& I \otimes...\otimes (\alpha_{x_{n}}^{k_{n}}\sigma_{x}+\alpha_{y_{n}}^{k_{n}}\sigma_{y}+\alpha_{z_{n}}^{k_{n}}\sigma_{z})]\|
\| e^{\| 2 \pi ij\sum_{l=1}^{n}k_{l}2^{-l}/N I^{\otimes N}\|} \nonumber \\
&\times&  e^{ \|i \theta [(\alpha_{x_{l}}^{k_{1}}\sigma_{x}+\alpha_{y_{l}}^{k_{1}}\sigma_{y}+\alpha_{z_{l}}^{k_{1}}\sigma_{z}) \otimes \cdots \otimes I+ \cdots I \otimes...\otimes (\alpha_{x_{n}}^{k_{n}}\sigma_{x}+\alpha_{y_{n}}^{k_{n}}\sigma_{y}+\alpha_{z_{n}}^{k_{n}}\sigma_{z})]\|}   \nonumber \\
&\leq& \frac{1}{\sqrt{N}} \sum_{j=0}^{N-1} 2^{n}\theta n\sqrt{2}e^{\theta n\sqrt{2}} \nonumber \\
&\leq& 2^{\frac{3n}{2}}\theta n\sqrt{2}e^{\theta n\sqrt{2}} \nonumber \\
&=&  2^{\frac{3\log N}{2}}\theta \log N\sqrt{2}e^{\theta \log N \sqrt{2}}.
\end{eqnarray}
\end{proof}

\section{Conclusions}

In this paper we present models of quantum neural networks using Clifford algebras $Cl(2n)$ and tensor product of Clifford algebras $Cl(3,0)^{\otimes n}$. Clifford architectures provided geometric and topological models of data. We show, within a rigorous mathematical framework, that this algebraic structure allows a general unitary learning algorithm (following the reference \cite{Shao}) and arbitrary activation functions implemented by the quantum-classical hibrid scheme. The basic idea is that unitary operators can be constructed from representations of Clifford algebras. The elements of Clifford algebras can be associated with geometric objects. The information related to orientation of subspaces can be encoded in multivectors which makes it possible controlling of subspaces without giving up information about their orientations \cite{Lounesto}. Therefore, our model is able to capture geometric information contained in the data and generate patterns that involve quantum entanglement. Hamiltonian simulation techniques allow an analysis of errors made if possible approximations are considered. An algebraic generalization of the quantum Fourier transform was proposed with additional parameters that enable its use as quantum machine learning models. Important properties such as unitarity, factorization and an upper bound for the distance between the standard quantum Fourier transform were derived. It is important to point out that applications of this generalized transform can go beyond quantum machine learning. An example circuit based on this algebraic formulation for two-qubit logic gates was presented, using two-level unitary quantum gates. This generalized approach make it possible to build specific models based on spinors and groups related to Clifford algebras such as the classical Clifford neurons \cite{Sven}. Another aspect to be highlighted is that since $Spin(n)$ is double cover of $SO(n)$ and $Lie(spin(n))\simeq so(n)$ we may to build representations of orthogonal neural networks using Clifford algebras in a systematic way. This networks may evade explosive or evanescent gradients \cite{Kere}. Finally, we believe that this proposal may be useful in the near-term quantum computing and as perspectives, we intend to implement our models in quantum computers.

\section{Appendix A}
In this appendix, we review some basic concepts about Clifford algebras \cite{Vaz, Lounesto}. \

 Given a vector space $V$, the Clifford algebra can be defined as quotient $Cl(V,Q)=\frac{T(V)}{I}$, where $I_{Q}$ is a two sided ideal generated by elements
 \begin{eqnarray}
 v \otimes v - Q(v)1,
 \end{eqnarray}
 for all $v \in V$; Q is the quadratic form and T(V) is the tensor algebra. Alternatively, let $V$ be a space vector over $\mathbb{R}$ equipped with symmetric bilinear form $g$, $A$ an associative algebra with unit $1_{A}$ and  $\gamma$ a linear application $\gamma:V\rightarrow A$. The pair $(A, \gamma)$ is a Clifford algebra for the quadratic space $(V,g)$ if $A$ is generated as an algebra by $\{\gamma(v); v \in V\}$, $\{a 1_{A}; a \in \mathbb{R}\}$ and satisfy
\begin{eqnarray}
\gamma(v)\gamma(u)+\gamma(u)\gamma(v)=2g(v,u)1_{A}, \label{one}
\end{eqnarray}
for all $v,u \in V$. Let $V$ be a vector space $\mathbb{R}^{n}$ and $g$ a symmetric bilinear form in $\mathbb{R}^{n}$ of signature $(p,q)$ with $p+q=n$.
We will denote by $Cl(p,q) \equiv Cl_{p,q} \equiv Cl(\mathbb{R}^{p,q})$ the Clifford algebra associated to quadratic space $\mathbb{R}^{p,q}$. In addition, we denote $Cl(2n,0)\equiv Cl(2n)$. The even subalgebra is defined by:
\begin{eqnarray}
Cl^{+}_{p,q}=\{\Gamma \in Cl_{p,q} ; \Gamma = \# \Gamma= \widehat{\Gamma} = (-1)^{p}\Gamma_{p}\},
\end{eqnarray}
 where $\#()$ and $\widehat{\ . \ }$ denote graded involution, that keep the sign of the elements belonging to even subspaces. 
The groups $spin(p,q)$ and $spin_{+}(p,q)$ are given by:
\begin{eqnarray}
spin(p,q)=\{a \in Cl_{p,q}^{+}; N(a)=\pm 1 \}
\end{eqnarray}
and
\begin{eqnarray}
spin_{+}(p,q)=\{a \in Cl_{p,q}^{+}; N(a)= 1 \},  \label{eq4}
\end{eqnarray}
respectively, where $N(a)=\mid a \mid ^{2}=\langle \widetilde{a} a \rangle_{0}$ is related to norm of elements of Clifford algebra, and $\widetilde{\ . \ }$ represent the reversion operator defined by
$\widetilde{A}_{[k]}=(-1)^{k(k-1)/2}A_{[k]};$   $\langle \ \ \rangle_{k}:Cl_{p,q}\rightarrow \bigwedge_{k}(\mathbb{R}_{p,q})$, with $\bigwedge_{k}(\mathbb{R}_{p,q})$ is the exterior algebra of vector space $\mathbb{R}_{p,q}$.  \\

The Lie Algebra of $spin_{+}(p,q)$ is the space of bivectors $\bigwedge^{2} \mathbb{R}^{p,q}$. Let $B_{1}$ and $B_{2}$ be two bivectors. Then the commutator
\begin{eqnarray}
[B_{1},B_{2}] \in \bigwedge^{2} \mathbb{R}^{p,q}
\end{eqnarray}
is a bivector.

\section{Appendix B}

Unitary operators can be obtained through exponentials of anti-Hermitian operators.  We will show how a simple circuit can be built from exponentials of elements of representations of Clifford algebras following the prescription obtained in the reference \cite{Nielsen}. Thus we will build a circuit related to the unitary transformation
\begin{eqnarray}
U(\theta_{1}, \theta_{2})&=&\exp[i(\theta_{1}(\sigma_{x} \otimes \sigma_{y}) + \theta_{2} (\sigma_{y}\otimes \sigma_{x})] \nonumber \\
&=& U( \theta_{1}) U( \theta_{2}) \nonumber
\end{eqnarray}
where
\begin{displaymath}
U(\theta_{1})=
\left(\begin{array}{cccc}
\cos(\theta_{1})& 0 & 0 & \sin(\theta_{1}) \\
0 & \cos(\theta_{1}) & \sin(\theta_{1}) & 0 \\
0 & -\sin(\theta_{1}) & \cos(\theta_{1}) & 0 \\
-\sin(\theta_{1}) & 0 & 0 & \cos(\theta_{1}) \\
\end{array} \right),
\end{displaymath}
and
\begin{displaymath}
U(\theta_{2})=
\left(\begin{array}{cccc}
\cos(\theta_{1})& 0 & 0 & \sin(\theta_{1}) \\
0 & \cos(\theta_{1}) & \sin(\theta_{1}) & 0 \\
0 & -\sin(\theta_{1}) & \cos(\theta) & 0 \\
-\sin(\theta_{1}) & 0 & 0 & \cos(\theta_{1}) \\
\end{array} \right).
\end{displaymath}
$U(\theta_{1})$ can be expressed as a product of two-level unitary gates
\begin{eqnarray}
U(\theta_{1})=U_{1}(\theta_{1})U_{2}(\theta_{1})U_{3}(\theta_{1}),
\end{eqnarray}
where
\begin{displaymath}
U_{1}(\theta_{1})=
\left(\begin{array}{cccc}
\cos(\theta_{1})& 0 & 0 & -\sin(\theta_{1}) \\
0 & 1 & 0 & 0 \\
0 & 0 & 1 & 0 \\
-\sin(\theta_{1}) & 0 & 0 & -\cos(\theta_{1}) \\
\end{array} \right),
\end{displaymath}

\begin{displaymath}
U_{2}(\theta_{1})=
\left(\begin{array}{cccc}
1& 0 & 0 & 0 \\
0 & \cos(\theta_{1}) & \sin(\theta_{1}) & 0 \\
0 & \sin(\theta_{1}) & -\cos(\theta) & 0 \\
0 & 0 & 0 & 1 \\
\end{array} \right)
\end{displaymath}
and
\begin{displaymath}
U_{3}(\theta_{1})=
\left(\begin{array}{cccc}
1& 0 & 0 & 0 \\
0 & 1 & 0 & 0 \\
0 &  0 & -1 & 0 \\
0 & 0 & 0 & -1 \\
\end{array} \right).
\end{displaymath}

Analogously, $U(\theta_{2})$ can be expressed as a product of two-level unitary gates
\begin{eqnarray}
U(\theta_{2})=U_{1}(\theta_{2})U_{2}(\theta_{2})U_{3}(\theta_{2})
\end{eqnarray}
where
\begin{displaymath}
U_{1}(\theta_{2})=
\left(\begin{array}{cccc}
\cos(\theta_{1})& 0 & 0 & -\sin(\theta_{1}) \\
0 & 1 & 0 & 0 \\
0 & 0 & 1 & 0 \\
-\sin(\theta_{1}) & 0 & 0 & -\cos(\theta_{1}) \\
\end{array} \right),
\end{displaymath}

\begin{displaymath}
U_{2}(\theta_{1})=
\left(\begin{array}{cccc}
1& 0 & 0 & 0 \\
0 & \cos(\theta_{1}) & -\sin(\theta_{1}) & 0 \\
0 & -\sin(\theta_{1}) & -\cos(\theta) & 0 \\
0 & 0 & 0 & 1 \\
\end{array} \right)
\end{displaymath}
and
\begin{displaymath}
U_{3}(\theta_{1})=
\left(\begin{array}{cccc}
1& 0 & 0 & 0 \\
0 & 1 & 0 & 0 \\
0 &  0 & -1 & 0 \\
0 & 0 & 0 & -1 \\
\end{array} \right),
\end{displaymath}
so that
\begin{eqnarray}
U(\theta)=U_{1}(\theta_{1})U_{2}(\theta_{1})U_{3}(\theta_{1})U_{1}(\theta_{2})U_{2}(\theta_{2})U_{3}(\theta_{2}).
\end{eqnarray}

The circuit is illustrated in figure below.

\ \

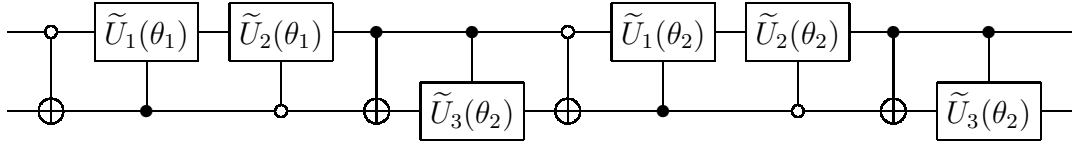
\Qcircuit @C=1em @R=.7em {
& \ctrlo{1} & \gate{\widetilde{U}_{1}(\theta_{1})} &  \gate{\widetilde{U}_{2}(\theta_{1})} & \ctrl{1} & \ctrl{1} & \ctrlo{1} & \gate{\widetilde{U}_{1}(\theta_{2})} &  \gate{\widetilde{U}_{2}(\theta_{2})} & \ctrl{1} & \ctrl{1} & \qw \\
& \targ     & \ctrl{-1}                            & \ctrlo{-1}                            &  \targ  &  \gate{\widetilde{U}_{3}(\theta_{2})} & \targ     & \ctrl{-1}                            & \ctrlo{-1}                            &  \targ  &  \gate{\widetilde{U}_{3}(\theta_{2})}& \qw
}

\ \

where
\begin{displaymath}
\widetilde{U}_{1}(\theta_{1})=
\left(\begin{array}{cc}
\cos(\theta_{1}) & -\sin(\theta_{1}) \\
-\sin(\theta_{1}) & -\cos(\theta_{1})  \\
\end{array} \right),
\widetilde{U}_{2}(\theta_{1})=
\left(\begin{array}{cc}
\cos(\theta_{1}) & \sin(\theta_{1}) \\
\sin(\theta_{1}) & -\cos(\theta_{1})  \\
\end{array} \right),
\widetilde{U}_{3}(\theta_{1})=
\left(\begin{array}{cc}
-1 & 0 \\
0 & -1  \\
\end{array} \right)
\end{displaymath}
and
\begin{displaymath}
\widetilde{U}_{1}(\theta_{2})=
\left(\begin{array}{cc}
\cos(\theta_{2}) & -\sin(\theta_{2}) \\
-\sin(\theta_{2} & -\cos(\theta_{2}) \\
\end{array} \right),
\widetilde{U}_{2}(\theta_{2})=
\left(\begin{array}{cc}
\cos(\theta_{2}) & -\sin(\theta_{2}) \\
-\sin(\theta_{2}) & -\cos(\theta_{2})  \\
\end{array} \right),
\widetilde{U}_{3}(\theta_{2})=
\left(\begin{array}{cc}
-1 & 0 \\
0 & -1  \\
\end{array} \right).
\end{displaymath}

Notice that
\begin{eqnarray}
U(\theta_{1})U(\theta_{2})\vert 00 \rangle = [\cos(\theta_{1})\cos(\theta_{2})-\sin(\theta_{1})\sin(\theta_{2})] \vert 00 \rangle + [\sin(\theta_{1})\cos(\theta_{2})-\sin{\theta_{2}}\cos{\theta_{1}}] \vert 11 \rangle, \nonumber \\
\end{eqnarray}
which is generally an entangled state. We will build states with the following notation
\begin{eqnarray}
\vert x; \theta \rangle &=& U(\theta)U(x)\vert 00 \rangle  \nonumber \\
&=&[\cos(\theta)\cos(x)-\sin(\theta)\sin(x)] \vert 00 \rangle + [\sin(\theta)\cos(x)-\sin(x)\cos(\theta)] \vert 11 \rangle \nonumber
\end{eqnarray}
where $\theta_{1}=\theta$ and $\theta_{2}=x$ and
\begin{eqnarray}
\vert y \rangle = U(y) \vert 00 \rangle = \cos(y) \vert 00 \rangle - \sin(y) \vert 11 \rangle. \nonumber
\end{eqnarray}


\end{document}